\documentclass[prl,aps,twocolumn,footinbib,showpacs]{revtex4}
\usepackage{bm,graphicx,amsmath,bbm}
\usepackage{amsmath,amsthm,amssymb,bbm}
\usepackage{txfonts}

\newtheorem{proposition}{Proposition}

\begin{document}

\title{Frustration, Entanglement, and Correlations in Quantum Many Body Systems}
\author{U. Marzolino$^1$, S. M. Giampaolo$^2$, and F. Illuminati$^{2,}$\footnote{Corresponding author: illuminati@sa.infn.it}}
\affiliation{$^1$\mbox{Physikalisches Institut, Albert-Ludwigs-Universit\"at Freiburg, Hermann-Herder-Stra\ss e 3, D-79104 Freiburg, Germany}
\\
$^2$\mbox{Dipartimento di Ingegneria Industriale, Universit\`a degli Studi di Salerno,~Via Ponte don Melillo, I-84084 Fisciano (SA), Italy}
}

\date{April 30, 2013}

\begin{abstract}

We derive an exact lower bound to a universal measure of frustration in degenerate ground states of quantum many-body systems. The bound results in the sum of two contributions: entanglement and classical correlations arising from local measurements. We show that average frustration properties are completely determined by the behavior of the maximally mixed ground state. We identify sufficient conditions for a quantum spin system to saturate the bound, and for models with twofold degeneracy we prove that average and local frustration coincide.
\end{abstract}


\maketitle


Many body systems are typically modeled by Hamiltonians that are sums of local terms.
Each local term operates only on a part of the entire system and acts to minimize the corresponding energy. If different subsystems overlap, the competition among the different local terms can preclude the existence of configurations satisfying all such minimizations simultaneously, a phenomenon known as frustration~\cite{Toulouse1977,Villain1977,Binder1986}. For classical Hamiltonian systems, frustration is associated to some nontrivial geometric property of the system itself~\cite{Binder1986}. On the other hand, due to quantum non-commutativity and entanglement~\cite{Nielsen2004,Acin2007}, classically unfrustrated systems may admit frustrated quantum counterparts~\cite{Wolf2003,Nielsen2004,Giampaolo2010,Eisert2010}. The existence at a qualitative level of a relation between frustration, non-commutativity, and entanglement has motivated recent efforts aiming at qualifying and quantifying frustration in the quantum domain. Frustration criteria, reducing to the Toulouse conditions in the classical case~\cite{Toulouse1977}, have been introduced recently for systems with non degenerate ground states~\cite{Giampaolo2011}, and a quantitative relation between a universal measure of frustration and ground-state bipartite entanglement has been established in the form of an exact inequality. This result cannot be extended unambiguously to degenerate ground states, because of the existence of infinitely many pure ground states with different entanglement properties. In addition, the possibility for mixed ground states introduces a further element, classical statistical correlations, in the quantification of frustration at the quantum level.

The present work introduces a theory of frustration for quantum systems with arbitrarily degenerate ground states. While including the one discussed in Ref.~\cite{Giampaolo2011} as a special case, it provides a general picture according to which the interplay between degeneracy and superselection rules results in a highly nontrivial relation between frustration and different types of correlations, classical and quantum. It provides a unified treatment and a rigorous quantification of frustration, expressed in terms of a universal inequality. The latter assesses the relative weights of geometry, bipartite entanglement, and shared classical correlations, identified as the three fundamental sources of frustration in quantum and classical systems. The unified approach is based on the concept of {\it maximally mixed ground state} (MMGS), that is the statistical mixture with equal a priori weights of all possible pure, degenerate ground states. This state always satisfies all the symmetries of the corresponding quantum Hamiltonian and therefore stands as the natural candidate to provide the relevant information on the global characteristics of a quantum system. Thus the concept of MMGS allows to introduce a complete classification of many-body quantum Hamiltonians, well beyond the limiting cases discussed in Ref.~\cite{Giampaolo2011}. Finally, we determine sufficient conditions for the saturation of the inequality that constitute a further quantum generalization of the classical Toulouse criteria.

Let us consider a many body Hamiltonian $H_T = \sum_{S}h_{S}$ sum of local Hamiltonians $h_{S}$. Frustration is the impossibility for the ground state $\rho_G$ of the total Hamiltonian $H_T$ to be entirely projected in the ground space of every local Hamiltonian $h_S$, that is the local ground space. To be more specific, let us consider a total system $T$, associated to the Hamiltonian $H_T$, bipartite into
a subsystem $S$, corresponding to the Hamiltonian $h_S$, and the remainder $R$. Define then $\rho_S={\rm tr}_R(\rho_G)$, the reduced ground-state density matrix of $S$, obtained by tracing out all the degrees of freedom of $R$ in $\rho_G$. Finally, introduce $\Pi_S$, the projector onto the local ground space of subsystem $S$. A natural way to quantify frustration is to consider the distance between $\rho_G$ and the projector onto the ground space of the local interactions~\cite{Giampaolo2011}:
\begin{equation} \label{f}
f_S=1-{\rm tr}\left(\rho_G \cdot \Pi_S \otimes 1\!\!1_R\right)=1-{\rm tr}_S\left(\rho_S \Pi_S\right) \, .
\end{equation}
\noindent
This quantity measures how much the global ground state fails to accommodate for the ground states of the local interactions.
Let us define $d$ as the rank of $\Pi_S$, namely the degeneracy of the local ground space associated to subsystem $S$. The Cauchy interlacing theorem~\cite{Bhatia1996} yields that the universal measure of frustration $f_S$ is bounded from below by the first $d$ largest eigenvalues
of $\rho_S$, arranged in descending order $\{\lambda_i^\downarrow(\rho_S)\}_i$~\cite{Giampaolo2011}:
\begin{equation} \label{bound}
f_S \geqslant \epsilon^{(d)}_{S} \, , \qquad \epsilon^{(d)}_{S} = 1 - \sum_{i=1}^d\lambda_i^\downarrow(\rho_S) \, .
\end{equation}

For non degenerate, pure ground states $\rho_G = |\Psi_G \rangle \langle \Psi_G|$, the quantity $\epsilon^{(d)}_{S}$ coincides with the distance of $\rho_G$ from the set of states with Schmidt rank less or equal than $d$~\cite{Turgut2010}. Such distance is not a full bipartite entanglement monotone: it vanishes for all separable states and is non-increasing under local operations and classical communication (LOCC)~\cite{Turgut2010}; however, it vanishes also for all entangled states with Schmidt rank $D_{Schmidt} \leq d$. It is then restored to a full entanglement monotone for all entangled states with $D_{Schmidt} > d$. Therefore, for $d=1$ the quantity $\epsilon^{(1)}_{S}$ is a faithful ground-state bipartite entanglement monotone, because for any entangled pure state $D_{Schmidt} \geq 2$. Indeed, $\epsilon^{(1)}_{S}$ coincides with the ground-state bipartite geometric entanglement, defined as the distance of $\rho_G$ from the set of bi-separable pure states $|\phi_S \rangle \otimes |\chi_R \rangle$~\cite{Blasone2008}. On the other hand, the fact that $\epsilon^{(d)}_{S}$ is not necessarily a full entanglement monotone when $d > 1$ reveals the subtlety of the relation between entanglement and frustration. States in a degenerate global ground space are entangled just because of the linearity of quantum mechanics, and yet this {\em a priori} entanglement
needs not be a source of frustration. An exquisite example is the ferromagnetic spin-$1/2$ Ising model on a triangle. This simple model is obviously frustration-free, either classically or quantum mechanically (all local interactions commute). Nevertheless, in the quantum regime, we have that any state of the form $|\psi\rangle = \alpha |\uparrow \uparrow \uparrow\rangle + \beta |\downarrow\downarrow\downarrow\rangle$
is an acceptable ground state that exhibits a nonvanishing entanglement for every possible bipartition as long as $\alpha,\beta\neq 0$. The global ground space is thus doubly degenerate. On the other hand, being the local ground space twofold degenerate as well, for the global ground state it will always be
$D_{Schmidt} \leq d = 2$. Regardless of the values of the superposition coefficients, this implies $\epsilon^{(2)}_{S} = 0$ and therefore the absence of entanglement-induced sources of frustration.

According to Eq.~(\ref{bound}), the ground state can belong to three different classes~\cite{Giampaolo2011}. It will be a {\em frustration free} (FF) state if \mbox{$f_S = \epsilon^{(d)}_{S}=0 \; \forall S$}. It will be an {\em inequality saturating} (INES) state if \mbox{$f_{S} = \epsilon^{(d)}_{S}\; \forall S$}; therefore, an FF state is a particular case of an INES state. Finally, it will be a {\em non inequality saturate state} (non-INES) in all the other situations. As long as the systems being considered admit only one non degenerate ground state, the same classification applies unambiguously also to the corresponding models.

New issues arise if we consider systems possessing different degenerate ground states. In this case the quantification of frustration provided by Eq.(\ref{f}) will yield in general state-dependent results within the same class of models and symmetries. For each model frustration becomes a {\em local}, state-dependent concept, according to the value taken by Eq.(\ref{f}) on each different degenerate ground
state. On one side, this feature confirms that entanglement is a necessary but not sufficient ingredient to characterize and quantify frustration {\em globally} in the quantum domain. On the other hand, it implies that the measure defined in Eq.(\ref{f}) should not be applied separately to each degenerate pure ground state. Rather, it should be evaluated on an appropriate, {\it average} ground state so defined that it contains all the possible information about the global ground space of the system. In this way it will be possible to quantify the {\em global} frustration properties of a model Hamiltonian rather than just its {\em local} features in the different degenerate ground states.

In the following, we will show that identifying the elements, beyond entanglement, that are needed to characterize the global aspects of frustration in the quantum domain, also determines the requirements for a state-independent, global quantification of frustration. To this end, the crucial observation is that statistical mixtures of degenerate ground states are themselves legitimate ground states. Moreover, when
symmetries are conserved, all the different degenerate ground states have the same statistical possibility to be realized. Therefore, one can introduce as the appropriate {\em global average} ground state the {\em maximally mixed ground state} (MMGS), that is the convex combination with equal a priori weights of all the possible degenerate ground states. This principle of {\it a priori} equiprobability guarantees the correct quantification of the global frustration properties of quantum Hamiltonians. The MMGS acts as a
projector on the global ground space and, at variance with single pure ground states, it satisfies all symmetries of the Hamiltonian model being considered. We can thus classify the different models and their global frustration properties with respect to the properties of the MMGS: a model is frustration free on average if its MMGS is frustration free; it is INES on average if its MMGS is INES; and it is otherwise non-INES on average if its MMGS is non-INES. If a system admits a unique, non degenerate ground state, then the local and on average classifications coincide. In all other cases if a model is locally FF (INES) then it is also FF (INES) on average, while in general the inverse does not necessarily hold.

When addressing mixed states, the entanglement-to-frustration relation undergoes essential modifications and generalizations. The quantity $\epsilon^{(d)}_{S}$ ceases in general to be a bipartite entanglement monotone when computed directly on mixed states~\cite{Adesso2003} and must be replaced by its convex roof
\begin{equation}
\label{conv_eps}
E^{(d)}_{S|R} = \inf_{\{p_k,|\psi_k\rangle\}}
\sum_k p_k \epsilon^{(d)}_{S}\left({\rm tr}_R|\psi_k\rangle\langle\psi_k|\right) \, .
\end{equation}
In Eq.~(\ref{conv_eps}) the infimum is taken over all the possible convex decompositions into pure states $|\psi_k\rangle\langle\psi_k|$ of the total system's ground state: $\rho_G=\sum_kp_k|\psi_k\rangle\langle\psi_k|$. The quantities $E^{(d)}_{S|R}$ and $\epsilon^{(d)}_{S}$
always coincide for pure states, and in general differ for mixed states. Physically, this difference comes about because the noise present in the reduced density matrix of a globally mixed state is associated not only to the presence of entanglement, as for pure states, but also to the presence of classical statistical correlations that emerge after local generalized measurements on one subsystem~\cite{Henderson2001,Ollivier2001,Horodecki2005}.

Total correlations (quantum plus classical) in mixed states are usually quantified in terms of the von Neumann entropy; however, the relation holds in general: to every entanglement monotone there corresponds a type of classical or quantum correlation emerging when a local generalized measurement is implemented~\cite{Koashi2004}. For our purposes, we need to evaluate the classical correlations between two subsystems, say $a$ and $b$, in terms of $\epsilon^{(d)}_{a}$. Although the function $\epsilon^{(d)}_{a}$ is not a full monotone, it is anyway a proper quantifier of local mixedness, and as such detects correlations
between subsystems. Consider then a composite systems (made of the subsystems $a$ and $b$) in a mixed state $\rho$, and its reduced state $\rho_a ={\rm tr}_b \rho$. A generalized measurement on $b$ is defined by a set of positive operators $\{M_b (x)\}$, such that $\sum_x M_b (x)=1$~\cite{NielsenChuang2000,Geometry2006}. The
measurement detects the result $x$ with probability $p(x) = {\rm tr}\left( \mathbbm{1}_a \otimes M_b (x)[\rho]\right)$, leaving the system in the state $\rho(x)=\frac{1}{p(x)}\mathbbm{1}_a \otimes M_b (x)[\rho]$. Replacing the von Neumann entropy with $\epsilon^{(d)}_\alpha$ in the general expression for classical correlations~\cite{Henderson2001,Ollivier2001}, one has:
\begin{equation} \label{class.corr}
C^{(d)}_{a|b} = \epsilon^{(d)}_{a} - \min_{\{ M_b (x) \}}\sum_x p(x)
\epsilon^{(d)}_{a}\left({\rm tr}_b \rho(x)\right) \, .
\end{equation}
The classical correlations $C^{(d)}_{a|b}$ are expressed as the difference between the total
correlations $\epsilon^{(d)}_{a}$, evaluated on the reduced state $\rho_a$, and the smallest convex combination of total correlations, obtained as the minimum over all possible local generalized measurements $\{ M_b (x) \}$ on subsystem $b$.
On the other hand, it is well known that any $n$-partite mixed quantum state $\rho$ can be obtained from an $(n+1)$-partite pure state by tracing out the degrees of freedom of an ancillary party $A$. This fact allows to generalize to the case of the geometric quantity $\epsilon^{(d)}_{S}$ the results originally obtained for the von Neumann entropy~\cite{Koashi2004}, according to the following

\noindent [{\bf Theorem 1 - Purification, entanglement, and classical correlations}]:
Given a pure tripartite state $|\psi_{SRA}\rangle\in\mathcal{H}_S\otimes\mathcal{H}_R\otimes\mathcal{H}_A$ and its reduced density matrices $\rho_{SR}={\rm tr}_A|\psi_{SRA}\rangle\langle\psi_{SRA}|$, $\rho_{SA}={\rm tr}_R|\psi_{SRA}\rangle\langle\psi_{SRA}|$ and $\rho_S={\rm tr}_{RA}|\psi_{SRA}\rangle\langle\psi_{SRA}|$, then:
\begin{equation}\label{theorem}
\epsilon^{(d)}_{S} = E^{(d)}_{S|R} + C^{(d)}_{S|A} \; ,
\end{equation}
that is, the total correlations expressed by the geometric quantity $\epsilon^{(d)}_{S}$ are the sum of the bipartite entanglement (convex roof) between $S$ and $R$ and the classical correlations between $S$ and $A$.
The proof of the theorem is provided in the supplementary material~\cite{supp}. We will apply Theorem 1 to the purification of the MMGS with the ancilla $A$, where the presence of classical correlations detected by local generalized measurements on $A$ is a consequence of the degeneracy of the ground state. Conceptually, the ancillary party $A$ can be thought of as a suitable quantum reservoir entangled with the bipartite system $(S|R)$, yielding the MMGS as the on average reduced equilibrium state~\cite{PopescuNaturePhysics2006}. Equivalently, $A$ can be seen as a quantum reference system, such that if one traces over $A$, every ground state of $(S|R)$ is equiprobable, as identified by a complete
set of superselection rules~\cite{BartlettRMP2007}.

Comparing Eqs.~(\ref{theorem}) and (\ref{bound}) we obtain a unified lower bound to frustration encompassing both the non degenerate and degenerate cases:
\begin{equation} \label{newbound}
f_S \geqslant E^{(d)}_{S|R} + C^{(d)}_{S|A} \; .
\end{equation}
The above exact inequality yields that, in the quantum domain, frustration is not only due to the underlying geometry, as for the classical case, and/or to entanglement, as for systems with non degenerate ground states. In general, it depends on the interplay of these two elements with a third source, namely statistical correlations established outside of the $d$-fold degenerate local ground space due to the degeneracy encoded in the MMGS. When the ground state is pure and non degenerate, the general bound Eq.~(\ref{newbound}) evaluated on the MMGS reduces to Eq.~(\ref{bound}), as the classical correlations $C^{(d)}_{S|A}$ vanish and the convex roof $E^{(d)}_{S|R}$ coincides with $\epsilon^{(d)}_{S}$.


Let us illustrate with a simple but nontrivial example the difference between the local and on average characterizations of frustration provided, respectively, by Eq.~(\ref{bound}) and Eq.~(\ref{newbound}). Consider a ring of five spins with periodic boundary conditions described by a ferromagnetic Ising Hamiltonian \mbox{$H=-\sum_{i=1}^5 S_i^x S_{i+1}^x$ }
(\mbox{$S_6^\alpha \equiv S_1^\alpha \; \forall \alpha=x,y,z$}). All local terms in the Hamiltonian commute: in this sense, the model is classical. The global and local ground states are both two-fold degenerate (spin-flip symmetry). Since every element of the global ground space is FF, the model is locally always FF, i.e. FF on each of the different degenerate global ground states, and hence FF on average,
that is FF on the MMGS. Next, let us modify the Hamiltonian from classical Ising to quantum  $XY$:
\mbox{$H=-\sum_{i=1}^5 (S_i^x S_{i+1}^x + \Delta S_i^y S_{i+1}^y)$ }, while the geometry of the system remains unchanged. One would now expect that frustration in the system should arise from the non-commutativity of the local terms in the $XY$ Hamiltonian (now the local ground
space degeneracy $d=1$). Accordingly, one verifies that the model is INES on average, i.e. the inequality Eq.~(\ref{newbound}) is saturated by the MMGS, with the actual values of the frustration measure $f_S$ and of the total correlations $\epsilon_{S}^{(1)}$ depending on the anisotropy $\Delta$. For instance, for $\Delta=0.1$, one has $f_S \equiv \epsilon_{S}^{(1)} \simeq 0.476$ for each of the five different
spin pair-interaction terms. However, regardless of the value of $\Delta$, such model is never locally INES. Indeed, given the doubly degenerate global ground space, let us pick e.g. the ground state that is eigenstate of the parity operator along the $x$ direction with eigenvalue $+1$. For this pure ground state, the measure of frustrations $f_S$ takes always the same values as for the MMGS (this is actually an interesting general property for all twofold degenerate ground states~\cite{supp}). However, the bipartite geometric entanglement $E_{S|R}^{(1)}$ is always well below the total correlations $\epsilon_{S}^{(1)}$. For instance, with $\Delta=0.1$, one has $E_{S|R}^{(1)} \simeq 0.001 \ll \epsilon_{S}^{(1)} \simeq 0.476$.
The impossibility to saturate inequality Eq.~(\ref{bound}) is related to the fact that the selected pure ground state breaks the symmetry of the model Hamiltonian. Modifying the geometry, e.g. by adding a direct antiferromagnetic interaction between the first and the third spin: $H'=S_1^x S_{3}^x + \Delta S_1^y S_{3}^y$, introduces a further, geometric source of frustration. In this case the MMGS no longer saturates inequality Eq.~(\ref{newbound}) and the system ceases to be INES, thus signaling the presence of geometric frustration.

From the above discussions it follows that it would be desirable to identify a set of conditions to detect {\em a priori} the frustration properties of the global ground space of a given model or class of models. These conditions should include as a particular case the ones previously determined for non degenerate ground states~\cite{Giampaolo2011}. Extending the Toulouse criteria~\cite{Toulouse1977} to the quantum domain, we need to identify a {\em prototype} model that is INES on average and then define a
group of local operations under which the property of being INES on average is preserved. We define the prototype model as
\\
\noindent[{\bf Prototype model}]: A quantum spin Hamiltonian of the type:
\begin{equation}\label{heisham}
H\!=\!\sum_{ij}h_{ij}\!=-\sum_{ij,\mu}\! J^\mu_{ij} S_i^\mu S_j^\mu \, ,
\end{equation}
is a {\em prototype} model if there exists at least one local ground space common to all pair local interactions $h_{ij}$ and every local coupling vector $\vec J_{ij}$ has non-negative components.
Having defined the prototype model, we state the following:

\noindent[{\bf Conjecture 1 - INES property and prototype models}]: {\em All prototype models are INES on average.}\\
\noindent [{\bf Conjecture 2 - INES property and local transformations}]: {\em Every model obtained from a prototype model by local
unitary operations on each spin and partial transposition on any arbitrary set of sites $\{K\}$ is still INES on average}.

It is evidently quite hard to prove these two conjectures in all generality. In the supplementary
material~\cite{supp}, we provide an analytical proof for the one-dimensional quantum $XY$ model in the thermodynamic limit, and strong numerical evidence obtained for more than $2\cdot10^{5}$ randomly generated models with exchange interactions on arbitrary random graphs with a total number of sites $N\le9$.
Preservation of the INES property on average under partial transposition is to be expected because this property is directly related to the absence of geometric frustration in the model. Nevertheless, the preservation of the INES property on average is far from trivial. Indeed, contrary to what happens under local unitary transformations, the properties of degenerate ground states of a given
Hamiltonian {\em before} and {\em after} partial transposition need not be related.
As a straightforward example, let us consider a spin-$1/2$ Heisenberg chain (open boundary conditions) of $N=4$ spins with homogeneous nearest-neighbor ferromagnetic couplings. The global ground state is five-fold degenerate while the local ground space is three-fold degenerate and for any couple of interacting spins we have $f_S = \epsilon_{S}^{(3)} =0$. Performing partial transposition on spins $2$ and $4$ (or, equivalently on spins $1$ and $3$) the original model maps in a antiferromagnetic Heisenberg chain possessing both a nondegenerate global and local ground state with $f_S = \epsilon_{S}^{(1)} \simeq 0.067$ for both the pairs of spins $(1,2)$ and $(3,4)$, while $f_S = \epsilon_{S}^{(1)} = 1/2$ for the pair of central spins $(2,3)$. The transformed model has very different ground-state properties compared to the initial one, and yet it remains INES.


In conclusion, we have derived an exact lower bound on a universal measure of frustration in the general case of degenerate ground states. The bound is expressed as the sum of two contributions, one due to bipartite ground-state entanglement and one due to bipartite classical correlations that are established after local generalized measurements on a quantum reservoir or a quantum reference frame. This further source of frustration adds to geometry and entanglement, yielding a rather complex structure that involves several fundamental concepts of quantum physics: entanglement of mixed ground states, classical correlations arising from quantum degeneracy, and the purification of ground states via quantum reservoirs or quantum reference frames.

We have showed that the frustration properties of quantum many-body Hamiltonians are encoded in the maximally mixed ground state (MMGS), that is the convex combination with equal coefficients of all the degenerate pure ground states. Given such on average, global classification, we have determined the sufficient conditions for a quantum spin system to achieve the bound, generalizing the results obtained in the case of models with non degenerate ground states~\cite{Giampaolo2011}. For systems with doubly degenerate ground states we have proved rigorously that local and average frustration coincide: all degenerate ground states, and therefore also the MMGS, exhibit the same frustration properties~\cite{supp}.

The fact that the residual classical correlations after local generalized measurements are identified as a novel source of frustration in quantum many-body systems may open interesting research insights concerning competitions among quantal and non-quantal aspects in Hamiltonian spin dynamics, as in the case of the anisotropic Heisenberg models in an external field, or competition between quantal dynamics and thermal
fluctuations. Future investigations should address the role that the presence/absence of the INES property on average actually plays in computational~\cite{Vidal2010}, information-theoretic~\cite{Lacorre,Lewenstein,Wolf2008,RMP2010}, and thermodynamic characterizations~\cite{Ramirez,Balents} of many body quantum systems. Experimental consequences might soon be derived and tested, as the controlled quantum simulation of magnetism, classical and quantum, already includes the first demonstration of antiferromagnetic spin chains with optical lattices~\cite{Greiner2011}, the probing of small frustrated Heisenberg spin systems by photonic simulators~\cite{Walther2011}, the realization of classically frustrated Ising spins with trapped ions~\cite{Monroe2010} and optical lattices~\cite{Sengstock2011}, up to the recent comprehensive proposals for the quantum simulation of large classes of quantum frustrated magnetism with ion crystals~\cite{Plenio2012} and color centers in diamond~\cite{Plenio2013} that might open the way to the precise verification of long-standing predictions on exotic phases of matter. In this respect, the use of tools inspired by quantum information science can lead to a broader and deeper understanding of collective quantum phenomena via long-range {\it ``entanglement patterns''}~\cite{Wen2010}, beyond the Landau-Ginzburg framework of symmetry breaking and local order parameters~\cite{Jiang2012}.

F.I. and S.M.G. acknowledge financial support through the FP7 STREP Project iQIT, Grant Agreement No. 270843.










\vspace{0.4cm}

\begin{center} {\bf SUPPLEMENTARY MATERIAL}
\end{center}

\vspace{0.2cm}

\section{Monogamy of correlations}

We begin by reviewing the consequences of the Cauchy interlacing theorem~\cite{Bhatia1996} with respect to convexity and concavity of entanglement monotones, either faithful or unfaithful.

\begin{proposition}[Concavity of the generalized geometric entanglement] \label{concave}
The generalized geometric entanglement $\epsilon^{(d)}$ is concave:
\begin{equation} \label{concavity}
\epsilon^{(d)}_{S}(\rho)\geqslant\sum_k p_k \epsilon^{(d)}_S(\rho_k) \, ,
\end{equation}
\noindent
provided $\rho=\sum_k p_k \rho_k \, .$
\end{proposition}
\begin{proof}
From Cauchy interlacing theorem we have that, for all $k$ and for all orthonormal bases $\{|i\rangle\}_i$,
\begin{equation}
\sum_{i=1}^d\langle i|\rho_k|i\rangle\leqslant\sum_{i=1}^d\lambda_i^\downarrow(\rho_k) \, . \nonumber
\end{equation}
If $\{|i\rangle\}_{i=1,\dots,d}$ are the eigenvectors corresponding to the first $d$ largest eigenvalues
$\{ \lambda_i^\downarrow \}_{i=1,\dots,d}$ of $\rho$, ordered in descending order, one has that:
\begin{eqnarray}
 \sum_{i=1}^d\lambda_i^\downarrow(\rho)&=&\sum_{i=1}^d\langle i|\rho|i\rangle \nonumber \\
 &= &\sum_k \sum_{i=1}^d p_k \langle i |\rho_k|i\rangle \nonumber \\
& \le & \sum_k \sum_{i=1}^d p_k \lambda_i^\downarrow(\rho_k) \, . \nonumber
\end{eqnarray}
\end{proof}

On the other hand, every bipartite mixed state $\rho_{SR}$ of a bipartite system $(S|R)$ can be decomposed into infinite convex combinations or pure states:

\begin{equation}
\label{convexdecomp}
\rho_{SR}=\sum_k p_k |\psi_k\rangle\langle\psi_k|, \qquad p_k\geqslant 0, \qquad \sum_k p_k=1 \; ,
\end{equation}
\noindent
and a faithful measure of the entanglement in $\rho_{SR}$ can be defined via the convex roof
construction over pure-state entanglement:

\begin{equation}
\label{geometricmonotone}
E^{(d)}_{S|R} = \inf\sum_k p_k \epsilon^{(d)}_S({\rm tr}_R|\psi_k\rangle\langle\psi_k|) \; ,
\end{equation}

\noindent
where the infimum is taken over all the possible convex decompositions into pure states,
Eq.~(\ref{geometricmonotone}).

The degree of mixedness in the reduced density matrix $\rho_S$ of a bipartite mixed state
$\rho_{SR}$ is directly related not only to the degree of bipartite entanglement of $\rho_{SR}$
but also to a set of classical statistical correlations, namely all the correlations that
are left after a local generalized measurement is performed on one party (subsystem) of the total system~\cite{Henderson2001,Ollivier2001}. Given a bipartite system $(S|A)$, a generalized measurement (POVM) on party $A$ is described by a set of positive operators $\{M_x\}$, such that $\sum_x M_x=1$ \cite{NielsenChuang2000}. In a bipartite mixed state $\rho_{SA}$, with subsystem $S$ in the reduced state $\rho_S={\rm tr}_A\rho_{SA}$, a generalized measurement on $A$ yields some result $x$ with probability \mbox{$p_x={\rm tr}\left[ (\mathbbm{1}_S\otimes M_x ) \rho_{SA} \right]$}, leaving the total system in the state $\rho_x=\frac{1}{p_x} \left( \mathbbm{1}_S \otimes M_x \right) \rho_{SA}$.

A {\it bona fide} measure of the classical correlations $C^{(d)}_{S|A}$ existing between the two subsystems of globally bipartite system, satisfying a minimal set of axioms, was defined as the difference between the total correlations in the global state $\rho_{SA}$, measured by the von Neumann entropy of the reduced state $\rho_S$ and the residual quantum correlations after a local generalized measurement on party $A$ \cite{Henderson2001,Ollivier2001}. Here we introduce an analogous definition, except for the crucial difference that the von Neumann entropy is replaced by the concave function $\epsilon^{(d)}_S$:

\begin{equation}
\label{classicalcorrelations}
C^{(d)}_{S|A}=\epsilon^{(d)}_S(\rho_S)-\min_{\{M_x\}}\sum_x p_x \epsilon^{(d)}_S({\rm tr}_A\rho_x) \; .
\end{equation}


As we will now show, this is still a {\it bona fide} measure of bipartite classical correlations that satisfies the Koashi-Winter monogamy relation originally proven for the entanglement of formation~\cite{Koashi2004}:

\begin{proposition}[Koashi-Winter monogamy of correlations] \label{monogamy.corr}
Given a pure three-partite state $|\psi_{SRA}\rangle\in\mathcal{H}_S\otimes\mathcal{H}_R\otimes\mathcal{H}_A$ and their reduced density
matrices $\rho_{SR}={\rm tr}_A|\psi_{SRA}\rangle\langle\psi_{SRA}|$, $\rho_{SA}={\rm tr}_R|\psi_{SRA}\rangle\langle\psi_{SRA}|$ and
$\rho_S={\rm tr}_{RA}|\psi_{SRA}\rangle\langle\psi_{SRA}|$, the following relation holds

\begin{equation}
\epsilon^{(d)}_S(\rho_S) = E^{(d)}_{S|R} + C^{(d)}_{S|A} \; .
\end{equation}

\end{proposition}

\begin{proof}
An analogous proposition was originally proved in~\cite{Koashi2004}, with the entanglement of formation replacing the geometric monotone of entanglement $E^{(d)}_{S|R}$ defined by Eq.~(\ref{geometricmonotone}), and the standard classical correlations introduced in Refs.~\cite{Henderson2001,Ollivier2001} replacing the classical correlations $C^{(d)}_{S|A}$ defined by Eq.~(\ref{classicalcorrelations}). All these differences reduce to the function $\epsilon^{(d)}_S$ replacing the von Neumann entropy, used throughout in Ref.~\cite{Koashi2004}. On the other hand, the only property of the von Neumann entropy which enters in the proof for the original proposition of Ref.~\cite{Koashi2004} is concavity. Since the function $\epsilon^{(d)}_S$ is concave, as proved in Proposition~\ref{concave}, it follows that Proposition \ref{monogamy.corr} holds.
\end{proof}

If the state $\rho_{SR}$ is pure, then the classical correlations $C^{(d)}_{S|A}$ vanish, and the total correlations $\epsilon^{(d)}_S$ reduce to the generalized geometric entanglement $E^{(d)}_{S|R}$. If $\rho_{SR}$ is mixed, there are many possible purifications $|\psi_{SRA}\rangle$ of which $\rho_{SR}$ is the reduced state upon tracing out the degrees of freedom of party $A$. This fact implies in particular that:

\begin{proposition} The classical correlations $C^{(d)}_{S|A}$ do not depend on the purification procedure. Equivalently, all purifications $|\psi_{SRA}\rangle$ coincide up to a local unitary transformation on party $A$.
\end{proposition}
\begin{proof}
In order to prove this property explicitly, one needs first to recall that, by construction,
the classical correlations $C^{(d)}_{S|A}$, just like entanglement, are invariant under local
unitary operations~\cite{Henderson2001}. Moreover, let us recall that purification of a mixed state $\rho_{SR}$ amounts to attaching an orthonormal basis $\{|a_A^{(j)}\rangle\}_j$ to its eigenvectors $\{|v_{SR}^{(j)}\rangle\}_j$ as follows:

\begin{eqnarray}
& & \rho_{SR}=\sum_j p_j|v_{SR}^{(j)}\rangle\langle v_{SR}^{(j)}|\rightarrow \nonumber \\
& & \rightarrow\sum_j \sqrt{p_j}|v_{SR}^{(j)}\rangle|a_A^{(j)}\rangle=|\psi_{SRA}\rangle \, .
\end{eqnarray}

\noindent
If $\{|a_A^{(j)}\rangle\}_j$ were not an orthonormal basis, the reduced state would not coincide with $\rho_{SR}$. If $\rho_{SR}$ does not possess degenerate eigenspaces, then the only freedom is on the choice of the basis $\{|a_A^{(j)}\rangle\}_j$, which is equivalent to performing a local unitary operation on party $A$. If $\rho_{SR}$ possesses some degenerate eigenspace, there is freedom in the choice of its eigenvectors
$\{|v_{SR}^{(j)}\rangle\}_j$. One can purify the degenerate eigenspace ${\rm span}\{|v_{SR}^{(j)}\rangle \, | \, j\in\mathcal{P}\}$ of $\rho_{SR}$ with eigenvalues $p$, in the following ways:

\begin{eqnarray}
& &p \sum_{j\in\mathcal{P}} |v_{SR}^{(j)}\rangle\langle v_{SR}^{(j)}|=p \sum_{j\in\mathcal{P}} |\tilde v_{SR}^{(j)}\rangle\langle
\tilde v_{SR}^{(j)}|\rightarrow \nonumber \\
& & \rightarrow\sqrt{p} \sum_{j\in\mathcal{P}} |\tilde v_{SR}^{(j)}\rangle|a_A^{(j)}\rangle, \nonumber \\
& & |\tilde v_{SR}^{(j)}\rangle= \sum_{l\in\mathcal{P}} \xi_{j,l}|v_{SR}^{(l)}\rangle \, ,
\end{eqnarray}

\noindent
where $\xi_{j,l}$ are the matrix element of a unitary matrix, in order for $\{|\tilde v_{SR}^{(j)}\rangle\}$ to be orthogonal vectors. Thus

\begin{eqnarray}
& & \sqrt{p} \sum_{j\in\mathcal{P}} |\tilde v_{SR}^{(j)}\rangle|a_A^{(j)}\rangle=\sqrt{p} \sum_{l\in\mathcal{P}} |v_{SR}^{(l)}
\rangle|\tilde a_A^{(l)}\rangle, \nonumber \\
& & |\tilde a_A^{(l)}\rangle=\sum_{j\in\mathcal{P}} \xi^T_{l,j}|a_A^{(j)}\rangle \, ,
\end{eqnarray}

\noindent
and the freedom of purification reduces again to a unitary operation on the party $A$. This local unitary freedom can be considered as a part of the POVM $\{ M_x \}$. Since the minimization over all the POVMs enters in the definition Eq.~(\ref{classicalcorrelations}), the classical correlations $C^{(d)}_{S|A}$ are independent of the purification.
\end{proof}

\subsection{State-independent frustration in systems with doubly degenerate global ground states}

Consider spin-$1/2$ models described by Hamiltonians of the type
\begin{equation}
\label{heisenberg}
H\!=\!-\sum_{ij,\mu}\! J^\mu_{ij} S_i^\mu S_j^\mu \, ,
\end{equation}
where the $S_i^\mu$ operators are the usual spin operators along the $\mu$ direction acting on the $i$-th spin. Such Hamiltonians are invariant under the action of the parity operator $P_\alpha=\prod_j\sigma^\alpha_j$ (with $\alpha=x,y,z$), that is $[H,P_{\alpha}] = 0$. On the other hand, the parity operators along different directions do not commute when the system is made of an odd number of spins. This fact implies that the global ground state of $H$ cannot have degeneracy $d_G < 2$. In other words, the global ground state cannot be unique. In the class of Hamiltonians Eq.~(\ref{heisenberg}) we consider the case in which the global ground state is precisely two-fold degenerate. This instance is realized by the $XY$ model, either in a finite chain with an odd number of spins~\cite{Takahashi} or in the thermodynamic limit~\cite{Sachdev}. We emphasize that higher degeneracies ($d_G > 2$) may occur either for models such that there are strong resonances in the Hamiltonian parameters, and/or such that they obey additional symmetries, so that indeed twofold degeneracy ($d_G = 2$) holds typically. According to the aforementioned invariance under parity transformations, we can introduce as a basis for the global ground space the two eigenstates of the parity operator along a fixed direction $\alpha$ with eigenvalues $\pm 1$. These two eigenstates are denoted, respectively, as the even and odd states $|+_\alpha\rangle$ and $|-_\alpha\rangle$. Up to a global phase, one has that $|\mp_\alpha\rangle = P_\beta|\pm_\alpha\rangle$, with $\beta \perp \alpha$. As a consequence, the measure of frustration $f_S$ coincides on both states for every spin pair. All pure global ground states can be written as $|\psi\rangle = a|+_\alpha\rangle + b|-_\alpha\rangle$, with $|a|^2+|b|^2=1$. Given a spin pair $S$ and $\Pi_S$ the projector on any Bell state defined on $S$, one has that
\begin{eqnarray} \label{psi}
{\rm tr}(|\psi\rangle\langle\psi|\Pi_S)\!\!\! & = &\!\!\!  {\rm tr}\!\left[\left(|a|^2|+_\alpha\rangle\langle+_\alpha|+|b|^2|-_\alpha\rangle\langle
-_\alpha|\right)\Pi_S\right]  \nonumber \\
& &\!\!\! + a^*b \, {\rm tr}(|+_\alpha\rangle\langle -_\alpha|\Pi_S)+{\rm c.c.} \; .
\end{eqnarray}
Recalling the invariance of the Hamiltonian under the parity operations and that $|-_\alpha\rangle=P_\beta|+_\alpha\rangle$, one has further that
\begin{equation}
{\rm tr}\left(|+_\alpha\rangle\langle+_\alpha|\Pi_S\right)={\rm tr}\left(|-_\alpha\rangle\langle-_\alpha|\Pi_S \right) \; .
\end{equation}
To evaluate terms in the r.h.s. of Eq.~\ref{psi} one can re-write the states $|\pm_\alpha\rangle$ in the relative state representation~\cite{everett1957} with respect to the Bell basis of spin pair $S$. Taking into account that every Bell state is an eigenstate of some parity operator, one has that
\begin{equation}
|\pm_\alpha\rangle\!\!  = \!\! \sum_{\varepsilon = \pm 1}\big(c^\pm_\varepsilon|\psi_\varepsilon\rangle_S|\Psi^\pm_\varepsilon\rangle_R
+ d^\pm_\varepsilon|\varphi_\varepsilon\rangle_S|\Phi^\pm_\varepsilon\rangle_R\big) \; ,
\end{equation}
where $\varepsilon$ runs on the possible eigenvalues $\{\pm 1\}$ of the parity operators $P_\alpha$ acting on subsystem $S$, $|\psi_\varepsilon\rangle_S$ and $|\varphi_\varepsilon\rangle_S$ are the two Bell states
of spin pair $S$ such that when $P_\alpha$ is applied on them it yields the eigenvalue $\varepsilon$, while $|\Psi^\pm_\varepsilon\rangle_R$ and $|\Phi^\pm_\varepsilon\rangle_R$ are states defined on the remainder $R$ of the entire system. Since $|\pm_\alpha\rangle$ are eigenstates of $P_\alpha$, also $|\Psi^\pm_\varepsilon\rangle_R$ and $|\Phi^\pm_\varepsilon\rangle_R$ must be eigenstates of
the same parity operator when acting on the Hilbert space of $R$. Moreover, because $|+_\alpha\rangle$ and $|-_\alpha\rangle$ yield two different values of the parity, we must have that
$|\Psi^+_\varepsilon\rangle_R$ (or $|\Phi^+_\varepsilon\rangle_R$) must yield a different value of the parity than $|\Psi^-_\varepsilon\rangle_R$ (or $|\Phi^-_\varepsilon\rangle_R$). Collecting these simple facts it is straightforward to verify that
\begin{equation} \label{psi1}
{\rm tr}(|\psi\rangle\langle\psi|\Pi_S) = {\rm tr} \left(|+_\alpha\rangle\langle+_\alpha|\Pi_S\right) \; .
\end{equation}


Recalling the definition of the measure of frustration $f_S$, Eq.~(1) in the main text of the present letter~\cite{main}, from Eq.~(\ref{psi1}) it follows that the {\it local} frustration, that is $f_S$
evaluated in any degenerate pure ground state, does not depend on the chosen ground state in the case of a doubly degenerate ground space. Moreover, since the measure $f_S$ is linear, the {\it global}, on average frustration of the maximally mixed ground state (MMGS), that is the convex combination with equal weights of all the pure degenerate ground states, is the convex combination with equal weights of the {\it local} frustrations of the pure degenerate ground states. Since we have just seen that for doubly degenerate ground spaces the {\it local} measures of frustration coincide for all pure degenerate ground states, it follows that the {\it global} and {\it local} frustrations coincide.

\section{Evidence supporting Conjectures 1 - INES and prototype models; and Conjecture 2 - INES and partial transposition}

\subsection{Numerical results for random graphs with a finite number of spins}

Here we will discuss in detail the verification of the two Conjectures put forward in the main text of the present letter~\cite{main}, that is {\bf Conjecture 1}: {\em All prototype models are INES on average},
and {\bf Conjecture 2}: {\em Every model $H'=H^{\mathrm{T}_K}$ obtained from a prototype model $H$ by partial transposition on any arbitrary set of sites $\{K\}$ is still INES on average}.

In order to gain as much insight as possible on the validity of these two Conjectures, we have analyzed large numbers of randomly generated Hamiltonians without geometric frustration, and checked that they satisfy the equality $f_S = \epsilon_{S}^{(d)}$ for all pairs of spins (subsystems) $S$ entering in the local interaction terms $h_S$. Instances of the prototype model have been generated, where the gauge was fixed at the unitary sector of the gauge group, since this symmetry is trivially proven analytically. In order to check Conjecture 2, a random gauge transformation consisting of pure parity transformations was applied. In total
$2 \times 10^{5}$ geometric frustration-free Hamiltonians have been checked by exact diagonalization. No single instance was found that fails to satisfy $f_S = \epsilon_{S}^{(d)}$ for each pair of spins.

The numerical tests were performed according to the following route. We have generated random models of the $XYZ$ (general anisotropic gapped spin-$1/2$ models with exchange interactions), $XXZ$ (gapless Heisenberg) and $XXX$ (gapless isotropic Heisenberg), satisfying the conditions described in the main text~\cite{main}, either with homogeneous or inhomogeneous nearest-neighbor couplings. The models have been generated first by generating a connectivity graph $(\mathcal V,\mathcal E)$ with $|\mathcal V|=n$ vertices connected by edges $\mathcal E$. A spin-$1/2$ operator was placed at each vertex. For each edge $(ij)\in\mathcal E$, an interaction of the $XYZ$, or $XXZ$, or $XXX$ type was generated. Depending on whether an homogeneous or inhomogeneous model was produced, the interactions were either generated by replicating them identically for all edges, or all the single interactions were generated independently. Finally, either the identity or a parity transformation were randomly applied to every spin, each with probability $1/2$.
After having generated randomly both the lattice and the model Hamiltonian, we have diagonalized numerically the different Hamiltonians and for all pairs of nearest-neighbor interacting spins
we have determined the reduced density matrices and evaluated the measure of frustration $f_S$ and th geometric quantity $\epsilon_S^{(d)}$. All the numerical tests have been performed programming with MATHEMATICA 7.

In Tables~\ref{table1}, \ref{table2} and \ref{table3} we summarize the results of the extended numerical analysis, reporting the total number of models generated for each type of interaction, the number of {\em Accepted} Hamiltonians (those for which the INES condition $f_S = \epsilon_S^{(d)}$ is satisfied for all bonds
$(ij)\in\mathcal E$), and the number of {\em Rejected} Hamiltonians (those for which the INES condition $f_S = \epsilon_S^{(d)}$ is violated at least for some bond). The total number of spins (vertices) ranged from $5$ to $9$.The conclusion is that no rejected Hamiltonian has been found, meaning that no violation of Conjectures 1 and 2 could be identified.
\begin{table}[htbp]
\begin{center}
\begin{tabular}{ccc}
$n$~~& Accepted&Rejected\\
\hline
\hline
5&30000&0\\
\hline
7&20000&0\\
\hline
9&10000&0\\
\hline
\end{tabular}
\end{center}
\caption{\label{table1} Results for $XYZ$ models.}
\begin{center}
\begin{tabular}{ccc}
$n$& Accepted&Rejected \\
\hline
\hline
5&30000&0\\
\hline
7&20000&0\\
\hline
9&10000&0\\
\hline
\end{tabular}
\end{center}
\caption{\label{table2} Results for $XXZ$ models.}
\begin{center}
\begin{tabular}{ccc}
$n$& Accepted&Rejected\\
\hline
\hline
5&30000&0\\
\hline
7&20000&0\\
\hline
9&10000&0\\
\hline
\end{tabular}
\end{center}
\caption{\label{table3} Results for $XXX$ models.}
\end{table}%

\subsection{Analytical results for the $XY$ chain}

\begin{figure}[t]
\begin{center}
\includegraphics[width=7cm]{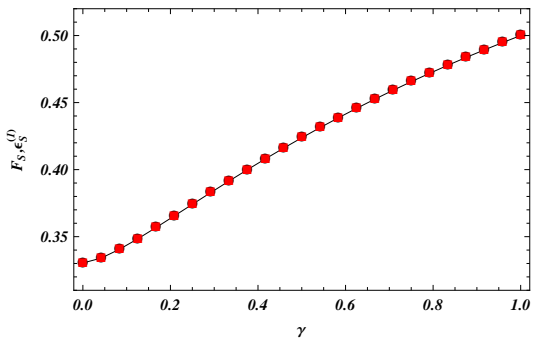}
\end{center}
\caption{Behavior of the bipartite geometric entanglement $\epsilon_S^{(1)}$ (black solid line) and of the universal measure of frustration $f_S$ (red dots) as functions of the anisotropy parameter $\gamma$ for a pair of neighboring sites in the $XY$ quantum spin-$1/2$ chain in the thermodynamic limit.}
\label{figxy}
\end{figure}

Our conjectures can be proved analytically for the case of the $XY$ quantum spin chain in the thermodynamic limit, exploiting the exact solution available for the model~\cite{Barouch}, whose Hamiltonian reads
\begin{equation}\label{hamxy}
H_{XY}\!=\!\frac{1}{2} \sum_{i}\! (1+\gamma) S_i^x S_{i+1}^x+ (1-\gamma) S_i^y S_{i+1}^y \; ,
\end{equation}
where $\gamma$ is the anisotropy parameter that can vary in the interval $[0,1]$, realizing the standard (classical) Ising model ($\gamma=1$), the generic (quantum and gapped) $XY$ model ($0<\gamma<1$), and the gapless, critical $XX$ model ($\gamma=0$). Except for the Ising case ($\gamma=1$) the local interactions between neighboring spins allow to accommodate a unique non degenerate local ground state while
the global ground state is doubly degenerate~\cite{Sachdev,Takahashi} except for the $XX$ point ($\gamma=0$) at which the global degeneracy is removed. Throughout, we consider the thermodynamic limit of $XY$ chains with periodic boundary conditions, so that translational invariance is enforced at all stages. In Fig.~\ref{figxy} we thus report the behavior of the bipartite geometric entanglement $\epsilon_S^{(1)}$ as a function of $\gamma$ (black solid line) for an arbitrarily chosen pair $S$ of neighboring spins and we compare it with the measure of frustration $f_S$ (red dots) for $0\le\gamma<1$. One sees that the two quantities coincide exactly for all values of the anisotropy parameter, proving exactly our two Conjectures
for this class of models. In the plot, the point $\gamma=1$ is excluded. In fact, at this point the local ground state becomes doubly degenerate and one must then consider the geometric quantity $\epsilon_S^{(2)}$ rather than $\epsilon_S^{(1)}$. At $\gamma=1$ the $XY$ chain turns in the classical, one-dimensional, frustration-free Ising model on a infinite chain. Indeed, it is straightforward to verify that both $\epsilon_S^{(2)}$ and $f_S$ vanish identically for this case.

\end{document}